\documentclass{article}
	
	\usepackage{amssymb}
	
	\usepackage{amsmath}
	\usepackage[pdftex]{graphicx}
	\usepackage{epsfig}
	
	\usepackage{framed}
	\usepackage{comment}
	\usepackage{etoolbox}
	\usepackage{fnpct}
	\usepackage{overpic}
	\usepackage{url}
	
	\usepackage[font=small,labelfont=bf]{caption} 
	\usepackage{relsize} 
	\usepackage{setspace}
	
	\usepackage{pdfpages}
	\usepackage{enumitem}

	\usepackage{subcaption}
	\usepackage{graphicx}
	
	\usepackage[utf8]{inputenc}

	\newcommand{\diag}{{\rm diag}}

	\newcommand{\supp}{{\rm supp}\, }

	\def\hide #1 {}
	\long\def\longhide #1 {}

	\usepackage{amsthm}
	
	\theoremstyle{plain}
	\newtheorem{theorem}{Theorem}[section]
	\newtheorem*{theorem*}{Theorem}

	\newtheorem{corollary}[theorem]{Corollary}

	\theoremstyle{definition}
	\newtheorem{remark}[theorem]{Remark}

	\newtheorem{definition}[theorem]{Definition}

	\usepackage[backend=biber,style=alphabetic]{biblatex}
	\title{On Graph Uncertainty Principle and Eigenvector Delocalization}
	\author{
		Elizaveta Rebrova\\Princeton University\\ Email: elre@princeton.edu
		\and 
		Palina Salanevich\\Utrecht University\\ Email: p.salanevich@uu.nl
	}

\begin{document}
		
		\maketitle
		\date{}
		
		\begin{abstract}
			
			Uncertainty principles present an important theoretical tool in signal processing, as they provide limits on the time-frequency concentration of a signal. In many real-world applications the signal domain has a complicated irregular structure that can be described by a graph. 
			In this paper, we focus on the global uncertainty principle on graphs and propose new connections between the uncertainty bound for graph signals and graph eigenvectors delocalization. We also derive uncertainty bounds for random $d$-regular graphs and provide numerically efficient upper and lower approximations for the uncertainty bound on an arbitrary graph. 
			
		\end{abstract}

		\section{Introduction and background}
		In many modern signal processing applications, such as social networks and co-authorship studies, brain imaging, epidemiology, and traffic networks, data has complicated relational structures. Such data can be naturally represented by signals defined on a (possibly weighted) graph, edges of which represent the relations between data points. The emerging field of \emph{signal processing on graphs} merges algebraic graph theory and computational harmonic analysis with the purpose to build tools for processing and analyzing such signals on graphs. Even though the theory of signal processing on graphs is very young, it develops rapidly \cite{shuman2013emerging, ortega2018graph}. 
		
		One of the cornerstone results in the classical time-frequency signal processing and harmonic analysis is the \emph{uncertainty principle}. It limits the degree to which a function can be simultaneously localized in time and frequency. Generalizations of such a principle to the graph setting, that is, establishing a trade-off between the spread of a signal in the vertex domain and in the frequency domain, is an important problem in signal processing on graphs. In particular, graph uncertainty bounds are used to
		study sampling properties on graphs \cite{perraudin2018global, tsitsvero2016signals} and to analyse space-frequency and wavelet decompositions \cite{agaskar2013spectral, shuman2013vertex}.
		
		\emph{In this paper, we study how the uncertainty principle for signals on graphs is linked  with the structure of the graph eigenvectors. More specifically, we (a) show that more refined localization results for the graph eigenvectors lead to better estimates of the uncertainty bound, and (b) propose the ways to quantify uncertainty bounds for particular graphs numerically.}

		To formally define uncertainty bound on graphs, we start with reviewing some relevant definitions. Let $G = (V, E)$ be a connected graph with the set of vertices~$V$, set of edges~$E$, and let $A$ be its adjacency matrix. 
		For simplicity, let us assume that $G$ is undirected and thus $A$ is symmetric. All the discussion (with slight modifications) can be extended to the case of directed graphs. 
		
		Let $\Lambda = \{\lambda_j\}_{j = 1}^{|V|}$ be the multiset of eigenvalues of the graph adjacency matrix $A$. Since $A$ is real and symmetric, there exists an orthonormal basis $\Phi = \{\varphi_\ell\}_{\ell \in \Lambda}$ consisting of eigenvectors~$\varphi_\ell$ of~$A$. 
		By slight abuse of notation, we identify a basis $\Phi$ with the matrix $\Phi$ that has basis vectors as its columns. 
		In some other works, $\Phi$ is chosen to be an orthonormal eigenbasis of the (normalized) graph Laplacian, see for instance~\cite{stankovic2019vertex}.
		
		Consider a signal $f:V\to \mathbb{C}$ defined on the vertex set $V$ of the graph $G$. Treating $f$ as a vector in $\mathbb{C}^{|V|}$, one can use representation of $f$ in the orthonormal eigenbasis $\Phi$ to define the graph Fourier transform (see, for example,~\cite{shuman2013emerging} for undirected and in \cite{singh2016graph, chen2015signal} for directed graphs).
		
		\begin{definition}\label{def_graph_fourier} We define the \emph{graph Fourier transform} $\hat{f}$  of a signal ${f: V\to \mathbb{C}}$ as $$\hat{f}(\ell) = \langle f, \varphi_\ell\rangle = \sum_{v\in V} f(v)\varphi_\ell^*(v), ~ \ell \in \Lambda.$$
			Since $\Phi$ is an orthonormal basis, the \emph{inverse Fourier transform} is given by $f(v) = \sum_{\ell\in \Lambda} \hat{f}(\ell)\varphi_\ell(v),~v\in V$.
		\end{definition}
		
		There are multiple ways to define uncertainty principle. In the continuous setup when $f\in L^2(\mathbb{R})$, the \emph{Heisenberg uncertainty principle} is defined using the notion of time and frequency spread \cite{folland1997uncertainty}, and for the finite dimensional signals $f\in \mathbb{R}^N$, it uses the sizes of the supports of $f$ and $\hat{f}$. Both of these approaches can be generalized in the graph setup, leading to local and global graph uncertainty principles \cite{perraudin2018global}. 
		\subsection{Local uncertainty principle on graphs}
		In classical harmonic analysis, for $f\in L^2(\mathbb{R})$, the Heisenberg uncertainty principle states that $\Delta_t^2 \Delta_{\omega}^2\ge \frac{1}{4}$, where $\Delta_t^2 = \min_{t_0\in\mathbb{R}}\frac{1}{\Vert f\Vert_2^2}\int_{\mathbb{R}}(t-t_0)^2|f(t)|^2dt$ denotes the \emph{time spread} of $f$ and $\Delta_{\omega}^2 = \frac{1}{\Vert f\Vert_2^2}\int_{\mathbb{R}}\omega^2|\hat{f}(\omega)|^2\frac{dt}{2\pi}$ denotes the \emph{frequency spread} of $f$ \cite{folland1997uncertainty}.
		Following this idea, Agaskar and Lu defined graph spread $\Delta_g^2(f)$ and spectral spread $\Delta_s^2(f)$ for graph-based signals in a similar way~\cite{agaskar2011uncertainty, agaskar2013spectral}. 
		\begin{definition}
			For $f:V\to\mathbb{R}$, $f\ne 0$, the \emph{graph spread} is $$\Delta_g^2(f) = \min_{u_0\in V}\frac{1}{\Vert f\Vert _2^2} f^{T}P_{u_0}^2f,$$ where $P_{u_0} = \diag\{d(u_0, v)\}_{v\in V}$ and $d(\cdot,\cdot)$ is the graph distance.
			The \emph{spectral spread} of $f$ is defined as $$\Delta_s^2 (f)= \frac{1}{\Vert f\Vert _2^2}\sum_{\lambda \in \Lambda}\lambda |\hat{f}(\lambda)|^2.$$
		\end{definition}
		It has been shown that the \emph{feasibility region} ${\mathcal{D} = \{(\Delta_g^2(f), \Delta_s^2(f))\}_{f:V\to \mathbb{C}}}$ for a connected graph $G$ is a bounded convex set \cite{agaskar2013spectral}. Moreover, bounds on $\Delta_g^2(f) \Delta_s^2(f)$ has been established for trees in~\cite{agaskar2011uncertainty} and for full graphs, stars, and Erd\H{o}s-R\'{e}nyi random graphs in \cite{agaskar2013spectral}. 
		\subsection{Global uncertainty principle on graphs}
		In this paper we are going to focus on a different, \emph{global version of the uncertainty principle} that generalizes the one proposed by Donoho and Stark~\cite{donoho1989uncertainty}. They showed that a signal and its Fourier transform cannot be simultaneously sparse. More precisely,  $\Vert f\Vert _0\Vert \hat{f}\Vert _0\ge N$, where $\Vert f\Vert _0$ denotes the support size of~$f$. A similar bound for graph signals, however, does not hold for all graphs $G$ and depends not only on the graph structure, but also on the choice of $\Phi$. For example, it is easy to check that for a full graph $K_N$, one can choose an eigenbasis $\Phi$ so that there exists a signal $f$ with ${\Vert f\Vert _0\Vert \hat{f}\Vert _0 = 2}$, while for a different eigenbasis $\tilde{\Phi}$, $\Vert f\Vert _0\Vert \hat{f}\Vert _0\ge N$ for all~$f$. At the same time, for a cycle graph $C_N$, any choice of $\Phi$ leads to the uncertainty bound $\Vert f\Vert _0\Vert \hat{f}\Vert _0\ge N/2$.
		
		Given the dependency of the graph uncertainty bounds on the graph structure and the eigenbasis choice, this work is motivated by the following two questions:
		\begin{enumerate}[leftmargin=*]
			\item For which $G$ do we have that $\Vert f\Vert _0\Vert \hat{f}\Vert _0\ge h(|V|)$ for any choice of $\Phi$ and some function $h$, s.t. $h(n)\to_{n\to \infty} \infty$?
			\item For given $G$ and $\Phi$, how can one estimate the bound for $\Vert f\Vert _0\Vert \hat{f}\Vert _0$?
		\end{enumerate}
		
		The first question has been partially answered in \cite{teke2017uncertainty}. In particular, the following connection between the uncertainty bound on a graph $G$ and the properties of $\Phi$ has been shown. 
		\begin{theorem}[\cite{teke2017uncertainty}]\label{thm_glob_unc_inf_norm}
			Let $G = (V, E)$ be a weighted graph and $\Phi$ be an orthonormal basis of the eigenvectors of its adjacency matrix $A$. Then, for any graph signal~$f: V\to \mathbb{C}$, $$\Vert f\Vert _0 \Vert \hat{f}\Vert _0\ge \Vert \Phi\Vert _{\infty}^{-2}$$
			where $\displaystyle \Vert \Phi\Vert _{\infty} = \max_{\ell\in \Lambda,~i\in V}|\varphi_{\ell}(i)|$.
		\end{theorem}
		
		This bound is only tight if $\Phi$ consists of the elements with the same absolute values, although in some models it gives correct asymptotic dependence on the size of the graph. In Section \ref{sec_uncertainty_via_delocalization}, we establish a better link between graph uncertainty bound and the distribution of the eigenvectors coefficients. In Section~\ref{sec_d_regular_graphs}, we discuss the uncertainty bound on a well-studied model of random $d$-regular graphs. In Section~\ref{sec_numerical_bounds}, we propose efficiently computable lower and upper bounds for the global uncertainty bound on an arbitrary graph and test their performance on different random graph models.

		\section{Global graph uncertainty and eigenvector delocalization}\label{sec_uncertainty_via_delocalization}
		
		First, we give a generalization of Theorem~\ref{thm_glob_unc_inf_norm} that quantifies global uncertainty via the subblocks of the eigenvector matrix.
		
		\begin{theorem}\label{thm_glob_unc_Frob_norm}
			Let $G = (V, E)$ be a weighted graph and $\Phi$ be an orthonormal eigenbasis of its adjacency matrix $A$. Then, for any graph signal~$f: V\to \mathbb{C}$, 
			\begin{align*}
				\Vert f\Vert _0 \Vert \hat{f}\Vert _0\ge \min_{S\subset \Lambda, K\subset V}\left\lbrace |K| |S| \colon \Vert \Phi_{K\times S}\Vert _F\ge 1\right\rbrace,
			\end{align*}
			where  $\Phi_{K\times S}$ is the restriction of the matrix $\Phi$ to the rows indexed by $K$ and columns indexed by $S$, and $\Vert \cdot\Vert _F$ denotes the Frobenius norm.
		\end{theorem}
		\begin{proof}
			Without loss of generality, assume that the signal $f$ is normalized so that $\Vert f\Vert _2 = 1$. Then, since $\Phi$ forms an orthonormal basis $\Vert \hat{f}\Vert _2 = \Vert f\Vert _2 = 1$. Denote $K = \supp(f)$, $S = \supp(\hat{f})$. Using Cauchy-Schwartz inequality, we obtain
			\begin{equation*}
				\begin{split}
					\Vert \hat{f}\Vert _2^2 & =  \sum_{\ell\in S}\left|\sum_{i\in K} f(i)\varphi_{\ell}^*(i) \right|^2\le \sum_{\ell\in S}\left[\sum_{i\in K} |f(i)|^2 \sum_{i\in K}|\varphi_{\ell}^*(i)|^2 \right]\\ 
					& = \sum_{\ell\in S, i\in K}|\varphi_{\ell}(i)|^2 = \Vert \Phi_{K\times S}\Vert ^2_F.
				\end{split}
			\end{equation*}
			\noindent That is, $\Vert \Phi_{\supp(\hat{f})\times \supp(f)}\Vert _F\geq 1$ for any graph signal ${f:V\to \mathbb{C}}$. 
			The claim follows.
		\end{proof}
		
		Note that both Theorem~\ref{thm_glob_unc_inf_norm} and Theorem~\ref{thm_glob_unc_Frob_norm} link the global uncertainly bound on a graph to the distribution of the coefficients of its eigenvectors. Informally, good uncertainty bound requires the absolute values of an eigenvector entries to be of similar order. This property is also known as the \emph{delocalization} of (graph) eigenvectors. 
		
		The eigenvector delocalization is one of the central topics in modern random matrix theory. It is possible to show that certain classes of random matrices exhibit eigenvector universality: their eigenvector matrices are distributed similarly to those of the Gaussian Orthogonal Ensemble (i.e., like a sample of Haar measure on the orthogonal group $O(n)$), see, e.g., \cite{anderson2010introduction,o2016eigenvectors}. Various measures of distribution similarity were studied; among the most popular ones in non-asymptotic context are the following
		\begin{itemize}[leftmargin=*]
			\item \emph{Supremum delocalization} (studied in a variety of papers staring from \cite{erdHos2009semicircle})  gives a bound of the largest coordinates $\Vert\varphi_{\ell} \Vert_{\infty}$ of the eigenvectors in $\Phi$;
			\item \emph{No-gap delocalization} (first introduced in \cite{rudelson2016no}) ensures that any sufficiently large set of coordinates of an eigenvector carries a non-negligible portion of its Euclidean norm.
		\end{itemize}
		These two types of delocalization measure complementary properties of the distribution of an eigenvector coefficients, as one of them rules out peaks while the other one rules out gaps. It is easy to see that neither of these two properties implies the other. For example, presence of a few large outlier entries in the matrix $\Phi$ fails supremum delocalization but a no-gap localization can still hold.
		
		While Theorem~\ref{thm_glob_unc_inf_norm} gives a direct link between supremum delocalization and the uncertainty bound (see Corollary~\ref{cor_dense_d-regular_uncertainty} below), Theorem~\ref{thm_glob_unc_Frob_norm} can be used to deduce the uncertainty bound from other types of delocalization. 
		
		\begin{corollary}\label{cor: no-gap deloc} (Uncertainty from no-gap delocalization)
			Let $G = (V, E)$ be a graph. Suppose there exists a function $g\colon \mathbb{R}^2\to \mathbb{R}$, such that for any $\varepsilon \in (0,1)$ and any unit eigenvector $\varphi$ of $G$,
			$\Vert (\varphi)_K \Vert^2 \geq \varepsilon$ implies $\vert K\vert \ge g(\varepsilon, \vert V \vert)$ for any $K \subset V$. Then, for any signal $f$ defined on $G$,  $$\Vert f \Vert_0 \Vert \hat{f} \Vert_0 \geq \min_{S\subset \Lambda} g\left(\frac{1}{\vert S \vert}, \vert V \vert\right) \vert S \vert.$$
		\end{corollary}
		\begin{proof}
			Suppose that for $K\subset V$ and $S\subset\Lambda$, ${\Vert \Phi_{K\times S} \Vert_{F}^2\geq 1}$. Then there exists $\ell\in S$, such that $\Vert(\varphi_\ell)_K\Vert\geq\frac{1}{\vert S\vert}$. Choosing $\varepsilon = \frac{1}{\vert S \vert}$, we deduce that $\vert K\vert \geq  g\left(\frac{1}{\vert S \vert}, \vert V\vert \right)$, and the claim of the corollary follows from  Theorem~\ref{thm_glob_unc_Frob_norm}.
		\end{proof}
		
		Note that no-gap delocalization results are restricted to individual columns in the eigenvector matrix and do not address how many eigenvectors can be localized on the same index set. New non-asymptotic delocalization results addressing submatrices of the eigenvector matrix would lead to the tighter bounds for the uncertainty. An example of a delocalization of this type is \emph{box delocalization} we define in the following corollary.
		
		\begin{corollary}\label{cor: box deloc}(Uncertainty from box delocalization) 
			Suppose that the eigenvectors of a graph ${G = (V, E})$ have box delocalization. That is, for ${\varepsilon \in (0,1)}$, ${k\in \{1, \dots \vert V \vert\}}$, and any coordinate subset $K\subset V$ with $\vert K \vert \le k$, we have $\left\vert \left\lbrace \ell \in \Lambda \colon \Vert (\varphi_\ell)_K \Vert_2^2 \geq\varepsilon\right\rbrace \right\vert \leq L(k, \varepsilon)$, for some function $L$. Then, for any signal $f$ defined on $G$, $$\Vert f \Vert_0 \Vert \hat{f} \Vert_0 \geq \max_{k, \varepsilon} \left\lbrace k L(k, \varepsilon) \colon \varepsilon L(k, \varepsilon) < 1\right\rbrace.$$
		\end{corollary}
		\begin{proof}
			For each ${k\in \{1, \dots \vert V \vert\}}$, let us fix ${\varepsilon \in (0,1)}$, such that $\varepsilon L(k, \varepsilon) < 1$. Then, for each fixed  $K\subset V$ with $\vert K \vert = k$, the box delocalization assumption implies that for any $S\subset \Lambda$ with $\vert S \vert \le L(k, \varepsilon)$, we have $\Vert \Phi_{K\times S} \Vert_{F}^2<1$. Thus, $\Vert \Phi_{K\times S} \Vert_{F}^2\geq 1$ implies that $\vert S \vert \vert K \vert \geq kL(k, \varepsilon)$, and the claim follows from Theorem~\ref{thm_glob_unc_Frob_norm}.
		\end{proof}
		
		\section{Uncertainty bounds for $d$-regular graphs}\label{sec_d_regular_graphs}
		
		In this section, we focus on (random) $d$-regular graphs and show how some known (de)localization results can be applied to obtain uncertainty bounds. Let $G_{n,d}$ denote a random $d$-regular graph chosen uniformly from all $d$-regular graphs on $|V| = n$ vertices. Delocalization properties of $d$-regular graphs were extensively studied over the recent years, e.g., \cite{bauerschmidt2017local,litvak2019structure,huang2021spectrum}. 
		The following supremum delocalization result was established in \cite{huang2021spectrum}.
		\begin{theorem}[\cite{huang2021spectrum}]\label{sup_d_graphs}
			For each $d>3$, there exist a constant $\omega_d\in \mathbb{N}$ and $c\in (0,1)$, such that any eigenvector $\varphi$ of the adjacency matrix of $G_{n,d}$ satisfies
			$$
			\|\varphi\|_{\infty} \le \frac{\log^{\Omega(1)} n}{\sqrt{n}} \|\varphi\|_2
			$$
			with probability at least $ 1 - O\left( n^{-(1-c)\omega_d}\right)$.
		\end{theorem}
		
		Using this result and Theorem \ref{thm_glob_unc_inf_norm}, we obtain the following uncertainty principle for $d$-regular graphs.
		
		\begin{corollary}\label{cor_dense_d-regular_uncertainty}
			For each $d>3$, the following uncertainty bound holds for any signal $f$ on $G_{n, d}$ with high probability
			$$
			\|f\|_0\|\hat f\|_0 \ge \frac{n}{\log^{\Omega(1)}n}.
			$$
		\end{corollary}
		
		\begin{remark}
			We note that Theorem~\ref{sup_d_graphs} indicates strong supremum delocalization of random $d$-regular graphs with high probability. More precisely, the maximal coordinate of the eigenvectors has the same, up to a poly-log factor, absolute value as its average coordinate. For this reason, more refined Theorem~\ref{thm_glob_unc_Frob_norm} used together with different type delocalization results would not allow to significantly improve the $n$-dependence in the obtained uncertainty bound for $G_{n,d}$.
			
		\end{remark}

		The following theorem shows that, while most of the $d$-regular graphs have nice delocalization properties, there exist infinite families of such graphs with localized eigenvectors. 
		
		\begin{theorem}[\cite{alon2021high}]\label{localized_dreg}
			For every $d = p + 1$, with $p$ prime, and parameter $\alpha\in (0, 1/6)$,  there are infinitely many integers $n$, such that there exists a $(d+1)$-regular graph $G_n$ on $n$ vertices with the following property. There are sets $T_n\subset V_n$ with $\vert T_n \vert = O(n^\alpha)$ and $S_n\subset \Lambda_n$ with $\vert S_n \vert = \lfloor\alpha\log_d(n) \rfloor$, such that $\supp(\varphi_\ell)\subset T_n$ for all $\ell \in S_n$.
		\end{theorem}
		
		By selecting $\alpha = \frac{1}{\log_d(n)}$, so that $\vert S_n\vert = 1$, we obtain that for any signal $f$ defined on a graph $G_n$ described in Theorem~\ref{localized_dreg}, $$
		\|f\|_0\|\hat f\|_0 \gtrsim n^{\frac{1}{\log_d(n)}} = d,
		$$
		and this bound is attained when $f = \varphi_\ell$ with $\ell\in S_n$. 
		We note that the proof of Theorem~\ref{localized_dreg} produces eigenvectors $\varphi_\ell$, $\ell\in S_n$ with the additional property of partial localization on the same index set.
		That is, these graphs also do not enjoy box delocalization.

		\section{Numerical bounds on graph uncertainty}\label{sec_numerical_bounds}
		Another approach to quantifying uncertainty via Theorem~\ref{thm_glob_unc_Frob_norm} is by directly estimating the quantity
		\begin{equation}\label{eq_uncertainty_bound_Frobenious}
			B = \min_{S\subset \Lambda,~K\subset V}\left\lbrace|K| |S| \colon \Vert \Phi_{K\times S}\Vert _F\ge 1\right\rbrace
		\end{equation}
		for a given (non-random) graph. It is a binary non-convex quadratic optimizartion problem that can be formulated as
		\begin{align}\label{opt_problem}
			\min\limits_{x_i, y_j} &\sum_{i,j=1}^n x_i y_j \nonumber \\
			\text{s.t. } &\sum_{i,j=1}^n \Phi^2_{ij}x_i y_j
			\ge 1, \\
			&x_i, y_j \in \{0,1\} \text{ for } i, j = 1,\ldots, n \nonumber\end{align}
		Note that the problem above can be reduced to $n^2$  decision problems of the existence of a submatrix of given size with Frobenious norm at least $1$. In turn, this is equivalent to the Weighted Edge Biclique Decision Problem on
		complete bipartite graphs that is known to be NP-complete \cite{pandey2020maximum}. Several recent papers propose approximate algorithms for searching for large bicliques under assumptions on the graph structure, including \cite{chen2022efficient, chen2021efficient}, and even approximate optimization is shown to be NP-hard in some cases \cite{tan2008inapproximability}. 
		
		Below, we employ some simple relaxation techniques to investigate empirically the dependence of $B$ on the graph size $n$ for different graph models. These bounds serve two goals: (a) they form hypotheses about the dependence $B(n)$; and (b) they provide a way to obtain  estimates for uncertainty bound for a given graph. 
		
		\begin{itemize}[leftmargin=*]
			\item \textbf{Supremum lower bound.} Clearly, the supremum-based lower bound of Theorem~\ref{thm_glob_unc_inf_norm} is very weak -- it would be tight only if all the entries of the matrix have the same absolute values. We still include it for comparison: for any $K\times S\subset V\times\Lambda$, we have $\Vert \Phi_{K\times S}\Vert _F^2 \le |K\Vert S\Vert |\Phi\Vert _{\infty}^2$. Thus, $$B \ge \Vert \Phi\Vert _{\infty}^{-2}.$$
			
			\item \textbf{Greedy lower bound.} Accounting for the sizes of the elements but not their locations, we obtain 
			\begin{equation*}\label{eq_lower_bound}
				B \ge \min \left\lbrace m \colon \sum_{j=1}^m \tilde\Phi^2(j)\ge 1\right\rbrace,
			\end{equation*} where $\tilde{\Phi}^2\in \mathbb{R}^{n^2}$ consists of the squared elements of $\Phi$ in descending order.
			\item \textbf{Linear lower bound.} A standard linear relaxation of a bilinear optimization problem on a bounded domain \cite{mccormick1976computability} leads to the following problem
			\begin{align}\label{opt_problem_linear_relax}
				\min\limits_{U, x, y} &\sum_{i,j=1}^n u_{ij} \nonumber \\
				\text{s.t. } &\sum_{i,j=1}^n \Phi^2_{ij} u_{ij}
				\ge 1, \; 0 \le x_i, y_j \le 1  \\
				& \max\{0, x_i + y_j - 1\} \le u_{ij}  \le \min\{x_i, y_j\}
				\nonumber
			\end{align}
			for the decision variables $U \in \mathbb{R}^{n \times n}$ and $x, y \in \mathbb{R}^n$.
			\item \textbf{Semidefinite lower bound.} 
			Our approach for an SDP relaxation follows \cite{anstreicher2009semidefinite}, after symmetrization of the problem~\eqref{opt_problem}, and leads to the following problem:
			\begin{align}\label{opt_problem_sdp_relax}
				\min\limits_{U, x} &\sum_{i,j=1}^{2n} J_{ij} u_{ij} \nonumber \\
				\text{s.t. } &\sum_{i,j=1}^{2n} \bar \Phi^2_{ij} u_{ij}
				\ge 1, \\
				& u_{ii} - x_i = 0, \; 0 \le x_i \le 1  \nonumber\\
				& \max\{0, x_i + x_j - 1\} \le u_{ij}  \le \min\{x_i, x_j\}
				\nonumber \\
				& U - x x^T \succcurlyeq 0  \nonumber
			\end{align}
			for the decision variables $U \in {\rm Sym}(2n \times 2n)$ and $x  \in \mathbb{R}^{2n}$.
			The last constraint is semidefinite by a standard Schur complements representation.  The matrices $J$ and $\bar \Phi^2$ are symmetrized versions of an all-ones matrix and the matrix $\Phi^2$, respectively, e.g., 
			$$\bar{\Phi} = 0.5 
			\begin{bmatrix} \textbf{0}_{n \times n} &\Phi\\
				\Phi^T & \textbf{0}_{n \times n}
			\end{bmatrix}, \; (\bar{\Phi}^2)_{i,j} = (\bar{\Phi}_{i,j})^2.
			$$
			\item \textbf{$k$-width upper bound.}  We upper-bound the optimal solution of \eqref{opt_problem} by a solution of the problem with additional constraint $$\max\left\{\sum x_i, \sum y_i\right\} \leq k.$$ In other words, we narrow down the set of submatrices in \eqref{eq_uncertainty_bound_Frobenious} to the submatrices with at most $k$ rows or columns and do an exhaustive search. For each possible subset of $k$ rows (columns), we take as many of the largest partial row-sums as needed to get at least $1$ in total.  The complexity of this problem is $O(kn^{k+2}\ln(n))$, so it is feasible only for small fixed $k$. Below, we take $k = 3$ to form an upper bound.
		\end{itemize}
		
		We consider two graph models to test the bounds described above: a $2 \lceil\sqrt{n}\rceil$-regular graph on $n$ vertices (Figure~\ref{fig:regular-graph}) and the same regular graph modified by adding high-degree vertices. More precisely, $\lceil 0.5n\rceil$ random edges are added to $\lceil 0.1 n\rceil$ randomly selected vertices (Figure~\ref{fig:modified-graph}). The first model satisfies theoretical delocalization results outlined in Section \ref{sec_d_regular_graphs} and we expect linear in $n$ (up to a logarithm) bound for uncertainty. The second model is constructed to disrupt the ``like-uniform" structure of the entries in the adjacency eigenvector matrix (see their order statistics in Figure~\ref{fig:eigenvec_decay}).
		We use it as an example where the uncertainty bound could grow sublinearly with the graph size. 
		
		Let us briefly discuss how these bounds address two goals mentioned above. Asymptotically, our bounds support the conjecture above, showing linear growth of $B(n)$ with $n$ in Figure~\ref{fig:regular-graph} and sublinear growth of $B(n)$ in Figure~\ref{fig:modified-graph} for the non-regular (modified) graph. For the individual uncertainty estimates, we see that there is still a considerable gap between the lower and upper bounds (illustrated as a cyan region in Figures~\ref{fig:regular-graph}~and~\ref{fig:modified-graph}), but more sophisticated relaxations give better lower bounds. In particular, a supremum-based lower bound, that is used for known asymptotics (see Section~\ref{sec_d_regular_graphs}), is not competitive on larger graphs. 
		
		Our code is available at \\\url{https://github.com/erebrova/uncertainty-delocalization}. 
		\begin{figure}[t]
			\centering\includegraphics[width=0.9\linewidth]{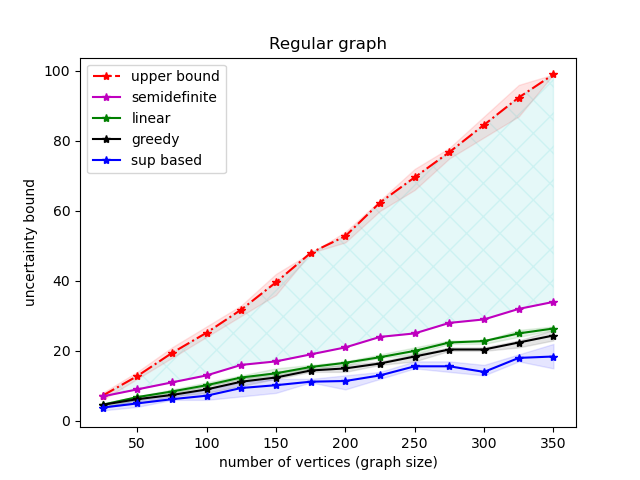} 
			\caption{Upper and lower estimates for the uncertainty bound \eqref{eq_uncertainty_bound_Frobenious} on $G_{n, 2 \lceil\sqrt{n}\rceil}$, averaged over $5$ runs (semidefinite lower bound from $1$ run). True value of $B$ is in the cyan region. All the bounds suggest linear growth of $B$ with $n$.}
			\label{fig:regular-graph}
		\end{figure}
		
		\begin{figure}[t]
			\centering\includegraphics[width=0.9\linewidth]{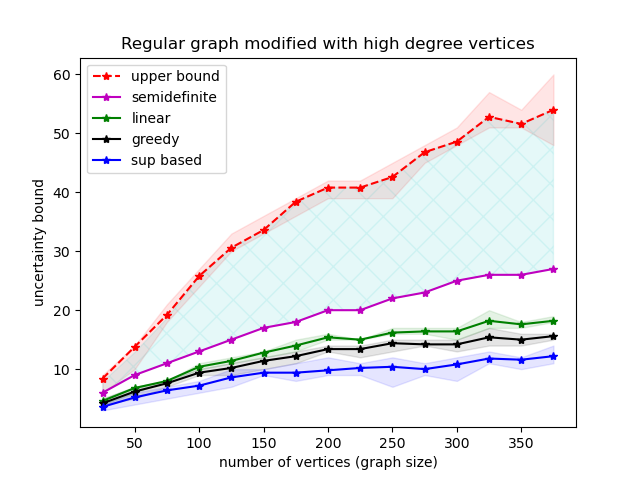} 
			\caption{Upper and lower estimates for the uncertainty bound \eqref{eq_uncertainty_bound_Frobenious} on $G_{n, 2 \lceil\sqrt{n}\rceil}$ modified by added high-degree vertices, averaged over $5$ runs (semidefinite lower bound from $1$ run). True value of $B$ is in the cyan region. The bounds suggest sublinear growth of $B$ with $n$.}
			\label{fig:modified-graph}
		\end{figure}
		
		\begin{figure}[h]
			\centering\includegraphics[width=0.9\linewidth]{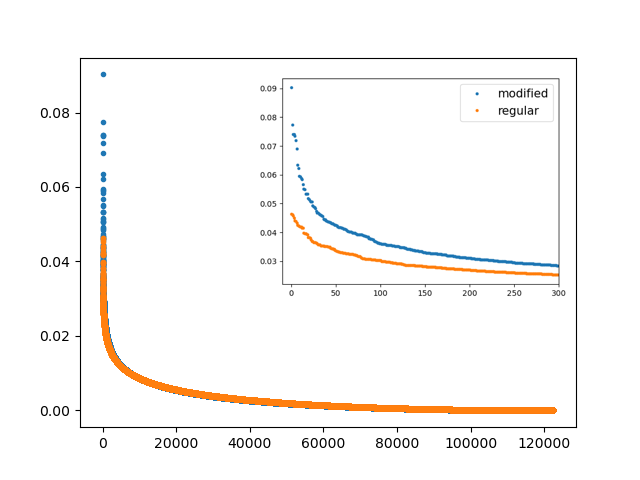} 
			\caption{Ordered entries of the entrywise squared absolute values of eigenvector matrices for $G_{n, 2 \lceil\sqrt{n}\rceil}$ and modified $G_{n, 2 \lceil\sqrt{n}\rceil}$ with ${n = 350}$ (upper right figure is zoomed on the top $300$ entries).}
			\label{fig:eigenvec_decay}
		\end{figure}

		\section{Discussion}
		The study of uncertainty principle is not only useful in graph signal processing applications, but it also stems a variety of interesting mathematical questions in spectral graph theory, optimization, and random matrix theory, and promotes connections between these diverse areas. In particular, the findings of this paper motivate the study of different, more refined, types of eigenvector delocalization that rule out simultaneous concentration of different eigenvectors on the same coordinate subsets (see, for instance, Corollary~\ref{cor: box deloc}). Moreover, study of the eigenvector delocalization for new families of (random) graphs with varying combinatorial properties, such as degree distribution and girth (e.g. \cite{slanina2012localization, huang2021spectrum}), would allow to establish a link between these properties and graph uncertainty.  
		
		In practice, an important problem is efficient numerical estimation of the uncertainty bound for a given graph. As discussed in Section~\ref{sec_numerical_bounds}, semidefinite relaxation comes with significant time requirements (the complexity is polynomial of large degree, as the input size is $n^2$, see~\eqref{opt_problem_sdp_relax}). Thus, for larger graphs, scalable approximation techniques (e.g., \cite{yurtsever2021scalable}) or dual reformulations could be considered in the future work.
		
		\section*{Acknowledgements}
		The authors are grateful to Bartolomeo Stellato and Ivan Kryven for inspiring and helpful discussions. Palina Salanevich is partially supported by NWO Talent programme Veni ENW grant, file number VI.Veni.212.176.

		\printbibliography
		
	\end{document}